\documentclass[10pt, conference, compsocconf]{IEEEtran}
\usepackage{etoolbox}

\ifCLASSINFOpdf
\else
\fi
\usepackage{graphicx}
\usepackage{cite}
\usepackage{xcolor}
\usepackage{soul,color}
\usepackage{tikz}
\usepackage{graphics}
\usetikzlibrary{backgrounds,calc,shadings,shapes,arrows,arrows.meta,shapes.arrows,shapes.symbols,shadows,patterns, datavisualization, datavisualization.formats.functions, matrix}
\usepackage{pgfplots}
\usepgfplotslibrary{groupplots}
\usepackage{dsfont}
\usepackage{amsmath}
\usepackage{amsthm}
\newtheorem{proposition}{Proposition}
\newtheorem{definition}{Definition}
\newtheorem{corollary}{Corollary}
\pgfplotsset{compat=1.14}

\usepackage[ruled,vlined]{algorithm2e}
\SetEndCharOfAlgoLine{}
\usepackage{subfig}

\begin{document}
%
\title{Identifying Operational Data-paths in Software Defined Networking Driven Data-planes}


\author{\IEEEauthorblockN{Jos\'{e} Reyes, Jorge L\'{o}pez, and Djamal Zeghlache}
\IEEEauthorblockA{SAMOVAR, CNRS\\
T\'{e}l\'{e}com SudParis / Universit\'{e} Paris-Saclay\\
9 Rue Charles Fourier, 91000 \'{E}vry, France\\
Email: \{jose.reyes,jorge.lopez,djamal.zeghlache\}@telecom-sudparis.eu}
}


%


\maketitle

\begin{abstract}
In this paper, we propose an approach that relies on distributed traffic generation and monitoring to identify the \emph{operational data-paths} in a given Software Defined Networking~(SDN) driven data-plane. We show that under certain assumptions, there exist necessary and sufficient conditions for formally guaranteeing that all operational data-paths are discovered using our approach. In order to provide reliable communication within the SDN driven data-planes, assuring that the implemented data-paths are the requested (and expected) ones is necessary. This requires discovering the actual operational (running) data-paths in the data-plane. In SDN, different applications may configure different coexisting data-paths, the resulting data-paths a specific network flow traverses may not be the intended ones. Furthermore, the SDN components may be defected or compromised. We focus on discovering the operational data-paths on SDN driven data-planes. However, the proposed approach is applicable to any data-plane where the operational data-paths must be verified and / or certified. A data-path discovery toolkit has been implemented. We describe the corresponding set of tools, and showcase the obtained experimental results that reveal inconsistencies in well-known SDN applications.

\end{abstract}


%
\IEEEpeerreviewmaketitle

\section{Introduction}\label{sec:intro}

Novel technologies allow flexible and fast network (re-)configuration of homogeneous network devices. Particularly, Software Defined Networking~(SDN) \cite{sdn} allows to centrally configure all data-plane (forwarding) devices; the data-plane devices (e.g., SDN-enabled switches) are configured by so-called SDN applications through the SDN controller. Some of the advantages of SDN are: (i) heterogeneous hardware can be managed with a single vendor-agnostic configuration interface; (ii) central configuration eases the management and reduces its execution time; and (iii) it avoids manual, error-prone configurations of the data-plane devices. For those reasons, SDN networks have evolved from small prototype networks to provider-scale network deployments \cite{googlesdn}, and their popularity constantly increases. Thus, guaranteeing the \emph{correct} functional and non-functional behavior of such systems is crucial \cite{enase18,ewdts18}. 
 
Data-plane devices are configured with \emph{flow rules}, that dictate the actions to perform once receiving the network packets. In fact, the highest priority rule that \emph{matches} a given packet is used to determine the action to take (e.g., drop, forward, etc.). However, as SDN networks are highly dynamic (forwarding devices can be frequently re-configured, and furthermore, by different applications), verifying that the packets follow the correct (intended) \emph{data-paths} is of special interest. Statically analyzing the rules installed in the data-plane is a common approach \cite{mai2011debugging}. However, by employing this approach it may be impossible to retrieve the operational data-paths configured in a given data-plane. For instance, when all network packets are sent to the SDN controller to query the action to take (and subsequently an SDN application decides on the appropriate action to perform, see Section~\ref{sec:sdn} for background concepts on SDN).

The data-paths installed in a data-plane must be \emph{correct} with respect to a number of functional and non-functional properties. For instance, from the functional standpoint the installed data-paths should coincide with the requested ones \cite{enase18}. From the non-functional standpoint (and particularly security), ensuring that there are no additional data-paths from the the intended ones is important to protect data secrecy; similarly, detecting fewer data-paths may be an indicator of a denial of service. In order to provide reliable communications within the SDN data-plane, it is necessary to verify that the data-paths configured in the data-plane are correct. However, to verify the data-paths configured at a given data-plane it is required to retrieve (discover) the data-paths that are \emph{actually implemented} in the data-plane. An immediate question follows: how can the data-paths implemented in a given SDN data-plane be discovered? Further, can it be (formally) guaranteed? This paper is devoted to reply to the previously stated questions, particularly, we focus on providing a formal methodology for retrieving the actual data-paths configured in the data-plane. We assume that there are no restrictions on the access for any Point of Control (PC) nor for any Point of Observation (PO), i.e., we assume we can stimulate the SDN network at any data-generation point, and likewise, that we can observe the network's reaction at any point. 

In the existing literature, there are few works that address the stated problem. Moreover, to the best of our knowledge no formal methods guaranteeing the data-path discovery have been presented (for more information see Section~\ref{sec:rel}). For that reason, we propose a distributed traffic generation and monitoring approach; different network packets are generated at selected nodes (hosts) of the data-plane, and then by monitoring the interfaces of the data forwarding devices, the traffic graph (or data-path) is retrieved. To guarantee all the data-paths installed in the data-plane are identified, we prove the conditions when the execution of a test suite (a set of test cases, i.e., network packets to be generated at given PCs) is necessary and sufficient to observe all implemented data-paths (Section~\ref{sec:analysis}). Using the proposed approach, a set of tools for data-path discovery has been implemented (Section~\ref{sec:tool}). We showcase the experimental results obtained by employing the developed tool. Particularly, we show how a wide-spread SDN application forwards data in an inconsistent and ineffective manner. Further, the developed tools may be used in different application areas. For example, guaranteeing the installed data-paths \emph{conform} to the requested ones; guaranteeing that the time to traverse a data-path is good (performance-wise); guaranteeing that there are no security \emph{faults} in the data-plane; etc.


\section{Background}\label{sec:background}
\input{TIKZ_network.tex}


\subsection{Software Defined Networking}\label{sec:sdn}
In traditional networks, the configuration, management, and data-forwarding interfaces are distributed / located at each of the data forwarding devices in the data-plane. The data-paths (the paths network packets follow in a data-plane) in the network are the result of the configuration on each of the forwarding devices; each of the devices has a local configuration and management interface. Thus, in order to re-configure the data-paths, several devices must be re-configured; as a consequence, while re-configuring each device the network may be in an inconsistent state, the process can be error-prone and slow. As an example, assume a data-plane in a traditional network as (only the data-plane) shown in Figure~\ref{fig:sdn_ex}. Assume there is an issue with the link between the switches $s2$ and $s3$ in the data-plane. The data-path depicted in dashed arrows ($h2\rightarrow s2\rightarrow s3\rightarrow h3$) becomes not operational. In order to re-configure this data-path, for example to $h2\rightarrow s2\rightarrow s4\rightarrow s3\rightarrow h3$, the switches $s2,s3$, and $s4$ must be re-configured, independently.

\begin{figure}[!htb]
    \centering
   \begin{tikzpicture}
        \node[server]                               (h1)    {};
        \node               at ([xshift=-1cm]h1)    (h1l)   {$h_1$};
        \node[l3 switch]    at ([xshift=2cm]h1)     (s1)    {};
        \node               at ([yshift=1cm]s1)     (s1l)   {$s_1$};
        \node[l3 switch]    at ([xshift=2cm]s1)     (s2)    {};
        \node               at ([yshift=1cm]s2)     (s2l)   {$s_2$};
        \node[server]       at ([xshift=2cm]s2)     (h2)    {};
        \node               at ([xshift=.7cm]h2)    (h2l)   {$h_2$};
        \node[l3 switch]    at ([yshift=-1.5cm]s1)  (s3)    {};
        \node               at ([yshift=-1cm]s3)    (s3l)   {$s_3$};
        \node[server]       at ([xshift=-2cm]s3)    (h3)    {};
        \node               at ([xshift=-1cm]h3)    (h3l)   {$h_3$};
        \node[l3 switch]    at ([xshift=2cm]s3)     (s4)    {};
        \node               at ([yshift=-1cm]s4)    (s4l)   {$s_4$};
        \node[server]       at ([xshift=2cm]s4)     (h4)    {};
        \node               at ([xshift=.7cm]h4)    (h4l)   {$h_4$};
        
        \draw[thick] (h1)--(s1);
        \draw[thick] (h2)--(s2);
        \draw[thick] (h4)--(s4);
        \draw[thick] (h3)--(s3);
        
        \draw[thick] (s1.east)--(s2.west);
        \draw[thick] (s3.east)--(s4.west);
        \draw[thick] (s1.south)--(s3.north);
        \draw[thick] (s2.south)--(s4.north);
        \draw[thick] (s2.south west)--(s3.north east);
        
        \node[diamond, fill=red!80!black, minimum size=0.15cm, inner sep=0] at (h1.west)  (pch1)  {};
        \node[diamond, fill=red!80!black, minimum size=0.15cm, inner sep=0] at (h3.west)  (pch3)  {};
        \node[diamond, fill=red!80!black, minimum size=0.15cm, inner sep=0] at (h2.east)  (pch2)  {};
        \node[diamond, fill=red!80!black, minimum size=0.15cm, inner sep=0] at (h4.east)  (pch4)  {};
        
        \node[circle, fill=green!80!black, minimum size=0.15cm, inner sep=0] at (s1.west)         (pos1w)  {};
        \node[circle, fill=green!80!black, minimum size=0.15cm, inner sep=0] at (s1.east)         (pos1e)  {};
        \node[circle, fill=green!80!black, minimum size=0.15cm, inner sep=0] at (s1.south)        (pos1s)  {};
        
        \node[circle, fill=green!80!black, minimum size=0.15cm, inner sep=0] at (s2.west)         (pos2w)  {};
        \node[circle, fill=green!80!black, minimum size=0.15cm, inner sep=0] at (s2.east)         (pos2e)  {};
        \node[circle, fill=green!80!black, minimum size=0.15cm, inner sep=0] at (s2.south)        (pos2s)  {};
        \node[circle, fill=green!80!black, minimum size=0.15cm, inner sep=0] at (s2.south west)   (pos2sw) {};
        
        \node[circle, fill=green!80!black, minimum size=0.15cm, inner sep=0] at (s3.west)         (pos3w)  {};
        \node[circle, fill=green!80!black, minimum size=0.15cm, inner sep=0] at (s3.east)         (pos3e)  {};
        \node[circle, fill=green!80!black, minimum size=0.15cm, inner sep=0] at (s3.north)        (pos3n)  {};
        \node[circle, fill=green!80!black, minimum size=0.15cm, inner sep=0] at (s3.north east)   (pos3nw) {};
        
        \node[circle, fill=green!80!black, minimum size=0.15cm, inner sep=0] at (s4.west)         (pos4w)  {};
        \node[circle, fill=green!80!black, minimum size=0.15cm, inner sep=0] at (s4.east)         (pos4e)  {};
        \node[circle, fill=green!80!black, minimum size=0.15cm, inner sep=0] at (s4.north)        (pos4n)  {};
        
        \node[rectangle,draw,dashed,fill=none,text depth=-3.5cm,text width=8cm, minimum width=8cm,minimum height=4cm] at ([xshift=.85cm, yshift=-.8cm]s1) (main)  {\textbf{Data-plane}};
        
        \draw[-Latex,very thick, gray!80, shorten <=0.2cm,shorten >=0.2cm]([yshift=0.15cm,xshift=0.3cm]h1.east)-- ([yshift=0.15cm]s1.west);
        \draw[-Latex,very thick, gray!80, shorten <=0.1cm,shorten >=0.1cm]([yshift=0.15cm,xshift=0.1cm]s1.east)-- ([yshift=0.15cm]s2.west);
        \draw[-Latex,very thick, gray!80, shorten <=0.3cm,shorten >=0.7cm]([yshift=0.15cm]s2.east)-- ([yshift=0.15cm]h2.west);
        
        \draw[-Latex,dashed, very thick, gray!80, shorten <=0.7cm,shorten >=0.3cm]([yshift=-0.15cm]h2.west)-- ([yshift=-0.15cm]s2.east);
         \draw[-Latex,dashed, very thick, gray!80, shorten <=0.1cm,shorten >=0.3cm]([yshift=-0.15cm]s2.south west)-- ([yshift=-0.15cm]s3.north east);
        \draw[-Latex,dashed, very thick, gray!80, shorten <=0.3cm,shorten >=0.5cm]([yshift=-0.15cm]s3.west)-- ([yshift=-0.15cm]h3.east);
        
        \node[server]       at ([yshift=2cm,xshift=1cm]s1)     (ctrl)    {};
        \node               at ([xshift=-2cm]ctrl)  (ctrll)   {SDN controller};
        
        \draw[thick, dotted, gray] (s1.north)--(ctrl.south east);
        \draw[thick, dotted, gray] (s2.north)--(ctrl.south west);
        \draw[thick, dotted, gray] ([xshift=0.1cm]s3.north east)--(ctrl.south);
        \draw[thick, dotted, gray] (s4.north west)--(ctrl.south);
        
        \node[rectangle,pattern=horizontal lines, minimum width=0.5cm, minimum height=1cm]    at([yshift=1.5cm, xshift=-2cm]ctrl)    (app1)  {};
        \node[] at ([xshift=-1cm]app1)    (app1l) {App. 1};
        \node[] at ([xshift=2cm]app1)   (dots)  {\huge\ldots};
        \node[rectangle,pattern=horizontal lines, minimum width=0.5cm, minimum height=1cm]    at([xshift=2cm]dots)    (app2)  {};
        \node[] at ([xshift=1cm]app2)    (app2l) {App. $N$};
        
        \draw[thick, dotted, gray] (app1.south)--(ctrl.north);
        \draw[thick, dotted, gray] (app2.south)--(ctrl.north);
   \end{tikzpicture}
    \caption{Example SDN infrastructure with data-paths, PCs and POs.}
    \label{fig:sdn_ex}
\end{figure}
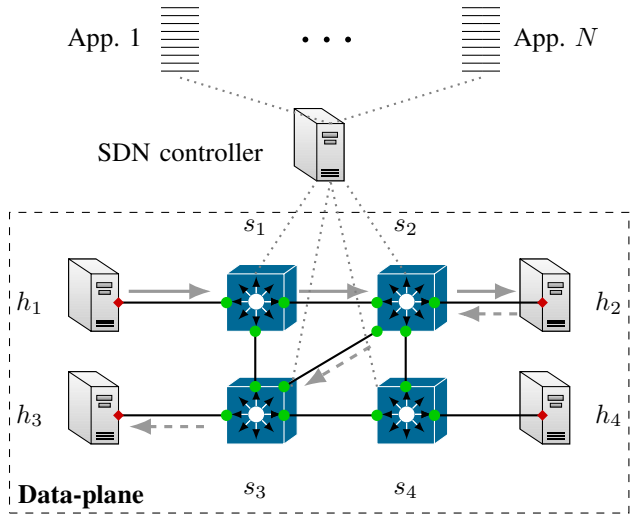

 SDN overcomes these limitations by separating the control and the data-plane layers \cite{sdn}. With a centralized SDN controller, SDN applications can automatically re-configure the SDN data-plane in a timely manner. Furthermore, the devices in the data-plane may have different protocols and interfaces (called southbound interface), while the controller has a single communication protocol (northbound interface) with the applications; thus, simplifying communication with heterogeneous and vendor-agnostic data-planes. Finally, SDN-enabled forwarding devices stir (route / forward) the incoming network packets based on so-called flow rules installed by the SDN applications (through the controller). A flow rule consists of three main (functional) parts: a packet matching part, an action part and a location / priority part. The matching part describes the values which a received network packet should have for a given rule to be applied. The action part states the required operations to perform to the matched network packets, while the location / priority part controls the hierarchy of the rules using tables and priorities. Finally, it is important to note that there exists a special output port for a flow rule, the controller port; when a packet is sent to the controller, the controller queries the SDN applications to decide the actions to perform to the packet; as a result, the controller may install new flow rules, drop or forward the packet to a specific port.
 
 
 \subsection{System testing}
 Traditionally, system testing is conceived as a procedure to guarantee that a given System Under Test~(SUT) is functionally and / or non-functionally correct. In this paper, we do not focus on guaranteeing correctness (we focus on providing means for guaranteeing such correctness), however, many of the concepts used in our approach are based on system testing and monitoring concepts. 
 
 Generally speaking, a given SUT $\mathcal{S}$ has an alphabet of input and output (observable reactions) symbols, $I$ and $O$, respectively. When $\mathcal{S}$ is stimulated with a sequence of inputs $\alpha \in I^*$, it produces a sequence of outputs $\beta \in O^*$. To guarantee the correct behavior of $\mathcal{S}$, formal / model-based testing for reactive systems is widely adopted. A test suite $TS$ is a set of input (test) sequences, sometimes a test suite can contain the expected output reactions. Additionally, a Point of Control~(PC) is an \emph{interface} where $\mathcal{S}$ can be stimulated (inputs can be generated for the system), and a Point of Observation~(PO), is an interface where the output reactions of $\mathcal{S}$ can be observed. Notice that we consider the distributed nature by stimulating $S$ at several PCs, and furthermore collecting the information from different POs. Nonetheless, the information is collected and processed at a centralized analyzer \cite{www}.

\section{Analyzing SDN data-planes}\label{sec:analysis}


\subsection{Basic concepts}

This section presents the main contribution of our work. First, we introduce the necessary definitions and assumptions, and later we prove the necessary and sufficient conditions when all the data-paths in a data-plane are discovered. In this paper, we consider the data-plane is the SUT (where existing data-paths need to be retrieved for later verification), and accordingly we make the following assumptions regarding its structure and functional aspects.

\begin{enumerate}
    \item The data-plane is given. We note that SDN controllers allow the automatic retrieval of the data-plane through the link layer discovery protocol \cite{lldp} or others. However, we assume that the discovery process is more reliable if it does not depend on information provided by the SDN network, i.e., the SUT. One of the reasons is that the discovery procedure may be used to estimate the correctness of the SUT \cite{enase18,ewdts18} (see Section~\ref{sec:rel} for related work).\label{ass:given}
    
    \item In the data-plane, there are data generation / reception devices, i.e., hosts, and data forwarding devices (e.g., switches); hosts do not forward traffic, and forwarding devices do not generate / receive data. Operating systems of networking devices may allow to violate this assumption, however, a networking device can be considered as two different ones in our model if this is the intended functional behavior. Additionally, the data-plane must have at least two hosts and one forwarding device; which makes sense from the functional point of view.\label{ass:nodestype}
    
    \item In the data-plane, devices are connected, however, no two hosts are connected to each other, and data forwarding devices have at least two connections. Furthermore, each pair of devices is connected though a single port, sharing a single link between them. These assumptions are reasonable for networking infrastructures, hosts are connected to a single \emph{access} forwarding device, and a forwarding device with a single connection cannot forward data. Likewise, when two devices are connected with many links, the link is usually \emph{bonded} and considered a single link with higher bandwidth. Finally, we assume that the communication between each pair of nodes is bidirectional.\label{ass:conntype} 
    
    \item In the data-plane, (network) packets generated by hosts are routed (stirred) by the forwarding devices. The forwarding devices can either \emph{drop} a packet (not forward it) or it can be forwarded to a set of output ports (and correspondingly neighboring devices). Furthermore, we assume the decision is taken only depending on the input port and network packet headers \cite{ofsspec}. Additionally, we assume the links are 100\% reliable (we do not consider packet loss), and the packets are not altered by the forwarding devices (although our approach may work independently from this assumption).\label{ass:traffic}
    
    \item In the data-plane, traffic can be generated at any of the hosts. Likewise, traffic can be observed at all the network interfaces in the data-plane. We assume there is no restriction of access. This is a reasonable assumption if the discovery procedure is part of a certification process. \label{ass:access}
    
    \item Finally, we assume that the configuration in the data-plane is not modified while the discovery procedure is being executed. This may sound as a strong assumption given the dynamic nature of SDN networks. However, this scenario is feasible for certain cases, for example for a certification process. Moreover, the packets generated for testing purposes are distinguishable from data passing though the data-plane.\label{ass:fixedconfig}
\end{enumerate}

Given the previously stated assumptions, we consider the data-plane is given in our approach using the following definition.

\begin{definition}
    A data-plane $D$ is a weighted graph $(V,E,\mathcal{I})$, where $V$ is the set of nodes, $E$ is the set of undirected edges (unordered pairs of nodes), and $\mathcal{I}$ is the \emph{interface} function, that maps an edge to a pair of pairs of nodes and interfaces (the port numbers at each node), i.e., $\mathcal{I}: E\mapsto (V\times\mathds{N})^2$ (we consider the set of natural numbers denoted by $\mathds{N}$ as the set of non-negative integers). Additionally, $D$ has the following properties:
    \begin{itemize}
        \item The set of vertices is the union of two subsets, the set of hosts $H\subset V$ and the set of forwarding devices $S\subset V$, which are disjoint, and their cardinalities are restricted in the following manner: $2\leq|H|<|V|$, $1\leq|S|<|V|$, and $|H|+|S|=|V|$.
        
        \item The connectivity of the graph holds the following properties: (i) $\forall h\in H\; deg(h)=1$; (ii) $\forall s\in S\; deg(s)\geq 2$, where $deg$ is the degree function, as usual; and (iii) $\forall e=\{v_a,v_b\} \in E\; (v_a\in H \implies v_b \not\in H) \wedge (v_b\in H \implies v_a \not\in H)$.
    \end{itemize}
\end{definition}

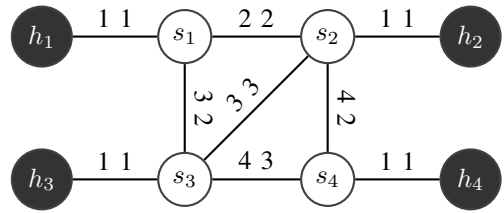
\begin{figure}[!htb]
    \centering
\begin{tikzpicture}[node distance=1.9cm,>=stealth',bend angle=45,auto]
    \tikzstyle{switch}=[circle,thick,draw=black!75,minimum size=6mm]
    \tikzstyle{host}=[circle,thick,draw=black!75,fill=black!80,minimum size=4mm,text=white]
    \tikzstyle{undirected}=[thick]
    \tikzstyle{directed}=[thick,->]
    
    \node[host] (h1) {$h_1$};
    \node[switch, right of=h1] (s1) {$s_1$};
    \node[switch, right of=s1] (s2) {$s_2$};
    \node[switch, below of=s1] (s3) {$s_3$};
    \node[switch, below of=s2] (s4) {$s_4$};
    \node[host, right of=s2] (h2) {$h_2$};
    \node[host, right of=s4] (h4) {$h_4$};
    \node[host, left of=s3] (h3) {$h_3$};
   
    \path   (h1)    edge[undirected]    node[above] {1 \hfill 1}   (s1)
            (h2)    edge[undirected]    node[above] {1 \hfill 1}   (s2)
            (h3)    edge[undirected]    node[above] {1 \hfill 1}   (s3)
            (h4)    edge[undirected]    node[above] {1 \hfill 1}   (s4)
            (s1)    edge[undirected]    node[above] {2 \hfill 2}   (s2)
                    edge[undirected]    node[above, rotate=-90] {3 \hfill 2}   (s3)
            (s2)    edge[undirected]    node[above, rotate=45] {3 \hfill 3}   (s3)
                    edge[undirected]    node[above, rotate=-90] {4 \hfill 2}   (s4)
            (s3)    edge[undirected]    node[above] {4 \hfill 3}   (s4);
\end{tikzpicture}
\caption{Example data-plane model}
    \label{fig:data-plane}
\end{figure}

As an example, consider the data-plane shown in Figure~\ref{fig:sdn_ex}, this data-plane is represented by the graph $D$ in Figure~\ref{fig:data-plane} with $H=\{h_1,h_2,h_3,h_4\}$, $S=\{s_1,s_2,s_3,s_4\}$, $E=\{\{h_1,s_1\},\{h_2,s_2\},\{h_3,s_3\},\{h_4,s_4\},\{s_1,s_2\},\{s_1,s_3\},\{s_2,s_3\},\{s_2,s_4\},\{s_3,s_4\}\}$, and $\mathcal{I}$ defined as follows: 
\begin{small}
\begin{align*} 
        \mathcal{I}(e)= &\begin{cases}
            ((h_1,1),(s_1,1)), & if\; e = \{h_1,s_1\} \\
            ((h_2,1),(s_2,1)), & if\; e = \{h_2,s_2\} \\ 
            ((h_3,1),(s_3,1)), & if\; e = \{h_3,s_3\} \\
            ((h_4,1),(s_4,1)), & if\; e = \{h_4,s_4\} \\
            ((s_1,2),(s_2,2)), & if\; e= \{s_1,s_2\}\\
            ((s_1,3),(s_3,2)), & if\; e= \{s_1,s_3\}\\
            ((s_2,3),(s_3,3)), & if\; e= \{s_2,s_3\}\\
            ((s_2,4),(s_4,2)), & if\; e= \{s_2,s_4\}\\
            ((s_3,4),(s_4,3)), & if\; e= \{s_3,s_4\}
        \end{cases} 
\end{align*}
\end{small}

According to our assumptions, and particularly with the node type and traffic assumptions (Assumption~\ref{ass:nodestype} and \ref{ass:traffic}, respectively), messages being sent from a host $h\in H$ are routed through different nodes of $S$, and the decision where to route is deterministic, and is based on the values within the message (header) and the interface (or input port) where the message is received. For that reason, we define the concept of \emph{traffic type}. First, we consider that a network packet $p$ is a binary string ($p\in\{0,1\}^*$), and thus so its header. A packet header has a predefined set of parameters. Therefore, we consider a packet header as a Boolean vector, denoted $\mathcal{H}(p)$. Consider the header with $n$ (relevant) parameters of total length $k=\sum_{i=1}^{n}k_i$:
\begin{small}
\[\mathcal{H}(p)=\underbrace{\overbrace{x_{1,1} x_{1,2} \ldots x_{1,k_1}}^{k_1}\overbrace{x_{2,1} x_{2,2} \ldots x_{2,k_2}}^{k_2}\ldots \overbrace{x_{n,1} x_{n,2} \ldots x_{n,k_n}}^{k_n}}_{k}\]
\end{small}

In the OpenFlow~(OF) protocol (a widespread and popular protocol for SDN-enabled switches) specification \cite{ofsspec}, there exists a minimum set of header parameters (and corresponding lengths) upon which the traffic can be matched (and later routed), for example, destination and source IP address, destination and source TCP port, etc. Therefore, forwarding devices route packets depending on the values of the Boolean vector $\mathcal{H}(p)$. In general, the matching part of a forwarding rule has the form ``$parameter\;1 = value\;1 \wedge parameter\;2 = value\;2 \wedge \ldots \wedge parameter\; n=value \;n$''; as an example, ``destination IP address = 10.0.0.1 $\wedge$ destination TCP port = 80''. Informally, we consider that the traffic type is a set of network packets whose header match particular values. Formally, we consider:

\begin{definition}
    A traffic type indicator $1_\tau$ is a Boolean function (characteristic function) for a given packet header $\mathcal{H}(p)\in \{0,1\}^k$  such that $1_\tau$ equals $1$ when $p$ belongs to the traffic type $\tau$.
\end{definition}

As an example, consider the traffic type $\tau=$``packets with header matching destination TCP port = 80'', assume the only two relevant parameters are destination IP address of length 32, and destination TCP port of length 16 (the corresponding Boolean vector has length 48). Considering that the first 32 inputs correspond to the destination IP address, and the next 16 to the destination TCP port, $1_\tau$ can be expressed as:
\begin{align*}
    \overline{x_{33}}\wedge\overline{x_{34}}\wedge\overline{x_{35}}\wedge\overline{x_{36}}\wedge\overline{x_{37}}\wedge\overline{x_{38}}\wedge\overline{x_{39}}\wedge\overline{x_{40}}\\\wedge\overline{x_{41}}\wedge x_{42}\wedge\overline{x_{43}}\wedge x_{44}\wedge\overline{x_{45}}\wedge\overline{x_{46}}\wedge\overline{x_{47}}\wedge\overline{x_{48}}
\end{align*}

A packet $p$ is of type $\tau$ if $1_\tau (\mathcal{H}(p)) = 1$; for convenience, we denote it as $p_\tau$. Further, the traffic of type $\tau$ is the set $\{p\in\{0,1\}^*|1_\tau(\mathcal{H}(p))=1\}$. It is important to note that a packet of type $\tau$ may be routed differently from a packet of type $\tau'$, coming from the same predecessor (node). That implies that the paths taken by a given packet $p$ in a data-plane $D$ depend on its traffic type. Likewise, packets of the same packet type may be routed differently, when they arrive from different predecessors; this recursively applies, until the origin of the packet (a host in the data-plane). When data are sent (encapsulated in network packets) by hosts in the data-plane (see assumption~\ref{ass:nodestype}), depending on the configuration of the forwarding devices, the data follow a specific \emph{data-path}.

\begin{definition}\label{def:data-path}
    A data-path for a packet $p$ with header $\mathcal{H}(p)$ in a data-plane $D=(H\cup S, E, \mathcal{I})$ starting at host $h\in H$, denoted $\pi_{D}(h, \mathcal{H}(p))$ is a non-empty (potentially infinite) sequence of directed edges that represents a path that a network packet with a specific header follows in the data-plane, for a given host and traffic type. The $j$-th edge in the path is denoted by $\pi_{D}^j(h, \mathcal{H}(p))$; we assume the first edge has index 1, and that the length of a data-path is the number of edges on it, denoted as $|\pi_{D}(h, \mathcal{H}(p))|$. We denote the set all data-paths for a packet $p$ with header $\mathcal{H}(p)$ starting at host $h$ as $\Pi_D(h,\mathcal{H}(p))$. Similarly, $\Pi_D(\mathcal{H}(p))$ the set of all data-paths for a packet $p$ with header $\mathcal{H}(p)$; the set of all data-paths of type $\tau$ in a data-plane is denoted as $\Pi_D(\tau)$; and the set of all data-paths in a data-plane is denoted as $\Pi_D$. We note that $\pi_{D}(h,\mathcal{H}(p))$ holds the following properties:
    \begin{itemize}
        \item $\forall j \in \{1,\ldots,|\pi_{D}(h, \mathcal{H}(p))|\}\; \pi_{D}^j(h, \mathcal{H}(p))=(v_a,v_b)\implies \{v_a,v_b\}\in E$, the directed edges are \emph{part} of $D$;
        \item $\forall j \in \{1,\ldots,|\pi_{D}(h, \mathcal{H}(p))| - 1\}\; \pi_{D}^j(h, \mathcal{H}(p))=(v_a,v_b)\implies \pi_{D}^{j+1}(h, \mathcal{H}(p))=(v_b,v_c)$, the sequence of edges connects vertices in $D$;
        \item $\pi_{D}^1(h, \mathcal{H}(p))=(h,s_h)$, data-paths start at hosts.
    \end{itemize}
\end{definition}

As an example, in Figure~\ref{fig:sdn_ex} we present two different data-paths $(h_1,s_1)(s_1,s_2)(s_2,h_2)$ (depicted with solid arrows), and $(h_2,s_2)(s_2,s_3)(s_3,h_3)$ (depicted with dashed arrows). It is important to note that a packet $p$ generated at host $h$ can induce data across a set of paths, the reason is that switches can \emph{clone} packets (forwarding traffic to a set of output ports, see assumption~\ref{ass:traffic}). However, a set of data-paths for a given packet has \emph{common prefixes}, for example, the set of paths $\{(h_1,s_1)(s_1,s_2)(s_1,s_3)(s_3,h_3),(h_1,s_1)(s_1,s_2)(s_2,h_2)\}$, reflects the fact that at $s_1$ when receiving a packet with header $\mathcal{H}(p)$ from $h_1$ the switch outputs the packet to $s_2$ and $s_3$ simultaneously. In a data-path a repeated edge implies the existence of a loop, and evidently, the data-path is of infinite length.

\subsection{Identifying operational data-paths}
In this section, we present a distributed test case generation / execution, and monitoring approach to discover the installed data-paths in a data-plane. We first give an overview of the proposed approach, and then proceed to the formal definitions and propositions that guarantee all (operational) data-paths are discovered using the presented approach. 

In order to discover the implemented data-paths on a data-plane, we propose to stimulate the SUT (the data-plane) at different points with network packets, and to monitor at different network interfaces in the data-plane, in order to observe the existing data-paths. We propose to install points of control at all hosts' interfaces, and points of observation at all forwarding devices' interfaces. As an example, in Figure~\ref{fig:sdn_ex} we illustrate the PCs as (red) diamonds, and the POs as (green) circles. We consider a distributed test case a specific PC (host) and a packet header; executing a test case in a data-plane is to send a packet with the given headers from the given (host) PC; as hosts have a single interface (see assumption~\ref{ass:conntype}), there is a single interface where to send the network packet. After a test case is executed against an SUT, the packet is forwarded and while traversing a network interface, a monitor installed at the interface in question sends the observed packets to a centralized processing service and the \emph{traces} are analyzed. We formally define:

\begin{definition}
    A distributed test case for a data-plane $D=(H\cup S, E, \mathcal{I})$, denoted as $tc_D$ is a pair $(h,\mathcal{H}(p))$, where $h\in H$, and $\mathcal{H}(p)$ is a network packet header. Correspondingly, a test suite $TS_D$ is a set of test cases for a given data-plane. Executing a test case $tc_D=(h,\mathcal{H}(p))$ means to generate $p$ (with header $\mathcal{H}(p)$) at host $h$.
\end{definition}

We note that given the traffic assumption (assumption~\ref{ass:traffic}), we consider that the links are 100\% reliable, and therefore, the test suites contain only distinct distributed test cases. Considering the reliability of a link in terms of probability is out of the scope of this paper and left for future work (see Section~\ref{sec:conc}). As discussed before, after executing a distributed test case on a data-plane, the generated network packet is observed though the monitoring interfaces. Formally, we define:

\begin{definition}
    An observation $\omega_{tc}$ after executing a test case $tc_D$ in a data-plane $D=(H\cup S, E, \mathcal{I})$ is a triple $(s,i,t)$, where $s\in S$, $i$ is the interface at the node $s$ where the packet is observed, and $t$ is a time-stamp, when the observation occurred. Correspondingly, the set of all observations after executing a given test case $tc_D$ is denoted as $\Omega_{tc}$. 
\end{definition}

As an example, consider the data-paths as depicted in Figure~\ref{fig:sdn_ex} (correspondingly, the data-plane as shown in Figure~\ref{fig:data-plane}). Assume the path $(h_1,s_1)(s_1,s_2)(s_2,h_2)$ (depicted with solid arrows) is configured for packets with header matching destination TCP port = 80. Executing a distributed test case $tc=(h_1,\mathcal{H}(p))$ (with header matching destination TCP port 80) produces the set of observations\footnote{For easiness and readability we use simple integers for the time-stamps.} $\Omega_{tc}=\{(s_1,1,1),(s_1,2,2),(s_2,2,3),(s_2,1,4)\}$. It is important to note that in this paper we assume all nodes in the data-plane are synchronized, and therefore, the time-stamps obtained at different nodes are comparable. Considering architectures with different clocks is out of the scope of the paper, and also left for future work. After understanding the relation between test cases and observations, an immediate question arises: can the set of observations be used to discover data-paths? For that reason, we present the following statement.

\begin{proposition}\label{stm:dpfromobs}
    For any set of observations $\Omega_{tc}$, after executing a test case $tc_D=(h,\mathcal{H}(p))$ in a data-plane $D=(H\cup S, E, \mathcal{I})$, the set of data-paths $\Pi_D(h,\mathcal{H}(p))$ can be computed.
\end{proposition}

\begin{proof}
   By construction (in Algorithm~\ref{algo:dpfromos}), we show that the set of observations can be used to compute a set of data-paths. Further, we note that by definition (Definition~\ref{def:data-path}) the set of data-paths that corresponds to the execution of $tc_D=(h,\mathcal{H}(p))$ is $\Pi_D(h,\mathcal{H}(p))$.
\end{proof}

\begin{algorithm}[!htb]
    \scriptsize
    \SetKwInOut{Input}{input}\SetKwInOut{Output}{output}\SetKw{KwBy}{by}
    \Input{A data-plane $D=(H\cup S, E, \mathcal{I})$, a test case $tc_D=(h,\mathcal{H}(p))$, the non-empty set of all observations after executing $tc_D$, $\Omega_{tc}$}
    \Output{A set of data-paths $\Pi_D(h,\mathcal{H}(p))$ or an error message}
    \textbf{Step 1:} Pop the observation with the smallest time from $\Omega_{tc}$ and assign it to $\omega=(v,i,t)$\;
    \textbf{Step 2:} Create a rooted tree $T$ with a single node, the root $R$, labelled with $h$, and set $TE$ of $R$ to $0$\;
    \textbf{Step 3:} \;
    \Indp\Repeat{$\Omega_{tc} \not= \emptyset$}
    {
        \textbf{Step 3.1:} Set the current level to the leafs of $T$\;
        \textbf{Step 3.2:}\;
        \Indp\uIf{$\exists$ a node $N$ in the current level (of $T$) that is labelled with $v$ (of $\omega$)}
        {
            \uIf{$N'$, the parent of $N$ (in $T$) is labelled by $s$, the node that is connected to $v$ though interface $i$ in $D$ and $TI$ is not set in $N$ and the $TE$ of $N'$ that corresponds to $N$ is less than $t$}
            {
                Set $TI$ of $N$ to $t$\;
            }
            \uElseIf{$TI$ is set in $N$ and $TI$ of $N<t$}
            {
                \If{Starting from $N$, recursively check if in the same branch (of $N$) it exists a node labeled by $s$ with a parent labelled by $v$}
                {
                    Error ($\exists$ a loop)!
                }
                Create a node $C$ labelled by $s$, the node that is connected to $v$ though interface $i$ in $D$\;
                Add $C$ to the children of $N$ and set $TE$ of $C$ in $N$ to $t$\;
                Set $TI$ of $N$ to $t$\;
            }
            \Else
            {
                Set the current level to the parents of the current level and go to \textbf{Step~3.2}\;
            }
        }
        \uElseIf{$\exists$ a node $N$ in the current level (of $T$) that is labelled with the node that is connected to $v$ though interface $i$ in $D$}
        {
            \If{Starting from $N$, recursively check if in the same branch (of $N$) it exists a node labeled by $s$ with a parent labelled by $v$}
            {
                    Error ($\exists$ a loop)!
            }
            Create a node $C$ labelled by $v$\;
            Add $C$ to the children of $N$ and set $TE$ of $C$ in $N$ to $t$\;
        }
        \Else
        {
            \If{the current level is not the root of $T$}
            {
                Set the current level to the parents of the current level and go to \textbf{Step~3.2}\;
            }
            \Else
            {
                Error (disconnected data-path)!
            }
        }\Indm
        \textbf{Step 3.3:} Pop the observation with the smallest time from $\Omega_{tc}$ and assign it to $\omega$\;
    }\Indm
    \textbf{Step 4:} Using a depth-first search, add to $\Pi_D(h,\mathcal{H}(p))$ all paths from the root to all leaves of $T$, and \Return{$\Pi_D(h,\mathcal{H}(p))$}
    \caption{Data-path construction from an observation set}\label{algo:dpfromos}
\end{algorithm}

Note that Algorithm~\ref{algo:dpfromos} creates a tree data structure (referred throughout the paper as a flow tree), in order to obtain the data-paths from the observations. Likewise, note that in the tree each of the nodes has a time where the message entered the node, the Time of Ingress (TI) and one Time of Egress (TE) per child, which denotes the time at which a message left the given node. Later, the set of paths is constructed using a simple depth-first search. To give a better intuition of how Algorithm~\ref{algo:dpfromos} works, consider a test case $tc_D=(h,\mathcal{H}(p))$, and a set of observations $\Omega_{tc}=\{(s_1,1,1),(s_1,2,2),(s_1,3,3),(s_2,2,4),(s_3,2,5),(s_2,1,6),(s_3,1,7)\}$; some of the stages of the \emph{flow tree} construction are depicted in Figure~\ref{fig:flowtree}. 

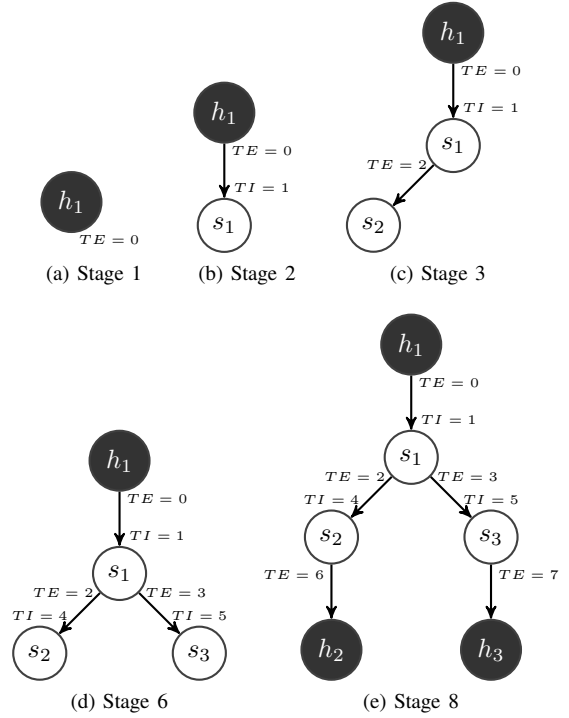
\begin{figure}[!htb]
    \centering
    \subfloat[t][Stage 1\label{fig:ft1}]
    {
    \begin{tikzpicture}[node distance=1.5cm,>=stealth',bend angle=45,auto]
    \tikzstyle{switch}=[circle,thick,draw=black!75,minimum size=6mm]
    \tikzstyle{sut}=[circle,thick,draw=black!75,pattern=dots,minimum size=6mm]
    \tikzstyle{host}=[circle,thick,draw=black!75,fill=black!80,minimum size=4mm,text=white]
    \tikzstyle{undirected}=[thick]
    \tikzstyle{directed}=[thick,->]
    \node[host] (h1) {$h_1$};
    \node[]  at([yshift=-0.5cm,xshift=0.5cm]h1)    (h1T){\tiny $TE=0$};
    \end{tikzpicture}
    }
    \quad
    \subfloat[t][Stage 2\label{fig:ft2}]
    {
    \begin{tikzpicture}[node distance=1.5cm,>=stealth',bend angle=45,auto]
    \tikzstyle{switch}=[circle,thick,draw=black!75,minimum size=6mm]
    \tikzstyle{sut}=[circle,thick,draw=black!75,pattern=dots,minimum size=6mm]
    \tikzstyle{host}=[circle,thick,draw=black!75,fill=black!80,minimum size=4mm,text=white]
    \tikzstyle{undirected}=[thick]
    \tikzstyle{directed}=[thick,->]
    
    \node[host] (h1) {$h_1$};
    \node[]  at([yshift=-0.5cm,xshift=0.5cm]h1)    (h1T){\tiny $TE=0$};
    \node[switch, below of=h1] (s1) {$s_1$};
    \node[] at([yshift=0.5cm,xshift=0.5cm]s1) (s1ti) {\tiny $TI=1$};
    
    \path   (h1)    edge[directed]    node[above] {}   (s1);
    
    \end{tikzpicture}
    }
    \quad
    \subfloat[t][Stage 3\label{fig:ft3}]
    {
    \begin{tikzpicture}[node distance=1.5cm,>=stealth',bend angle=45,auto]
    \tikzstyle{switch}=[circle,thick,draw=black!75,minimum size=6mm]
    \tikzstyle{sut}=[circle,thick,draw=black!75,pattern=dots,minimum size=6mm]
    \tikzstyle{host}=[circle,thick,draw=black!75,fill=black!80,minimum size=4mm,text=white]
    \tikzstyle{undirected}=[thick]
    \tikzstyle{directed}=[thick,->]
    
    \node[host] (h1) {$h_1$};
    \node[]     at([yshift=-0.5cm,xshift=0.5cm]h1)  (h1T){\tiny $TE=0$};
    \node[switch, below of=h1] (s1) {$s_1$};
    \node[]     at([yshift=0.5cm,xshift=0.5cm]s1)   (s1ti) {\tiny $TI=1$};
    \node[]     at([yshift=-0.25cm,xshift=-0.75cm]s1)  (s1te) {\tiny $TE=2$};
    \node[switch, below left of=s1] (s2) {$s_2$};
    
    \path   (h1)    edge[directed]  node[above] {}  (s1)
            (s1)    edge[directed]  node[above] {}  (s2);
    
    \end{tikzpicture}
    }
    \quad
    \subfloat[t][Stage 6\label{fig:ft6}]
    {
    \begin{tikzpicture}[node distance=1.5cm,>=stealth',bend angle=45,auto]
    \tikzstyle{switch}=[circle,thick,draw=black!75,minimum size=6mm]
    \tikzstyle{sut}=[circle,thick,draw=black!75,pattern=dots,minimum size=6mm]
    \tikzstyle{host}=[circle,thick,draw=black!75,fill=black!80,minimum size=4mm,text=white]
    \tikzstyle{undirected}=[thick]
    \tikzstyle{directed}=[thick,->]
    
    \node[host] (h1) {$h_1$};
    \node[]     at([yshift=-0.5cm,xshift=0.5cm]h1)  (h1T){\tiny $TE=0$};
    \node[switch, below of=h1] (s1) {$s_1$};
    \node[]     at([yshift=0.5cm,xshift=0.5cm]s1)   (s1ti) {\tiny $TI=1$};
    \node[]     at([yshift=-0.25cm,xshift=-0.75cm]s1)  (s1tes2) {\tiny $TE=2$};
    \node[]     at([yshift=-0.25cm,xshift=0.75cm]s1)  (s1tes3) {\tiny $TE=3$};
    \node[switch, below left of=s1] (s2) {$s_2$};
    \node[]     at([yshift=0.5cm]s2)   (swti) {\tiny $TI=4$};
    \node[switch, below right of=s1] (s3) {$s_3$};
    \node[]     at([yshift=0.5cm]s3)   (swti) {\tiny $TI=5$};
    
    \path   (h1)    edge[directed]  node[above] {}  (s1)
            (s1)    edge[directed]  node[above] {}  (s2)
            (s1)    edge[directed]  node[above] {}  (s3);
    
    \end{tikzpicture}
    }
    \subfloat[t][Stage 8\label{fig:ft8}]
    {
    \begin{tikzpicture}[node distance=1.5cm,>=stealth',bend angle=45,auto]
    \tikzstyle{switch}=[circle,thick,draw=black!75,minimum size=6mm]
    \tikzstyle{sut}=[circle,thick,draw=black!75,pattern=dots,minimum size=6mm]
    \tikzstyle{host}=[circle,thick,draw=black!75,fill=black!80,minimum size=4mm,text=white]
    \tikzstyle{undirected}=[thick]
    \tikzstyle{directed}=[thick,->]
    
    \node[host] (h1) {$h_1$};
    \node[]     at([yshift=-0.5cm,xshift=0.5cm]h1)  (h1T){\tiny $TE=0$};
    \node[switch, below of=h1] (s1) {$s_1$};
    \node[]     at([yshift=0.5cm,xshift=0.5cm]s1)   (s1ti) {\tiny $TI=1$};
    \node[]     at([yshift=-0.25cm,xshift=-0.75cm]s1)  (s1tes2) {\tiny $TE=2$};
    \node[]     at([yshift=-0.25cm,xshift=0.75cm]s1)  (s1tes3) {\tiny $TE=3$};
    \node[switch, below left of=s1] (s2) {$s_2$};
    \node[]     at([yshift=0.5cm]s2)   (swti) {\tiny $TI=4$};
    \node[switch, below right of=s1] (s3) {$s_3$};
    \node[]     at([yshift=0.5cm]s3)   (swti) {\tiny $TI=5$};
    \node[host, below of=s2] (h2) {$h_2$};
    \node[]     at([yshift=-0.5cm,xshift=-0.5cm]s2)  (s2tes6) {\tiny $TE=6$};
    \node[]     at([yshift=-0.5cm,xshift=0.5cm]s3)  (s3tes7) {\tiny $TE=7$};
    
    \node[host, below of=s3] (h3) {$h_3$};
    
    \path   (h1)    edge[directed]  node[above] {}  (s1)
            (s1)    edge[directed]  node[above] {}  (s2)
            (s1)    edge[directed]  node[above] {}  (s3)
            (s2)    edge[directed]  node[above] {}  (h2)
            (s3)    edge[directed]  node[above] {}  (h3)
            ;
    
    \end{tikzpicture}
    }
    \caption{Flow tree construction example}
    \label{fig:flowtree}
\end{figure}

Algorithm~\ref{algo:dpfromos} always terminates. There are three possibilities of how it terminates. First, if there is an observation that occurs at a given time, and in the flow tree there is no neighboring node with a smaller corresponding time, then Algorithm~\ref{algo:dpfromos} does not attach it to the tree and returns an error; the reason is that a data-path must be connected (according to its definition), and in our assumptions all interfaces are monitored. Identifying a disconnected data-path may imply that a P.O. in a preceding forwarding device does not report the packet to the processing server. This may be an indication of an attack, however, determining such issues is out of the scope of this paper. The second case occurs when a given observation occurs at a node that has been previously traversed, and that eventually produces an infinite loop. Although a data-path can be infinite (and theoretically observations too), from the practical standpoint it serves no purpose to report such data-paths. Therefore, when a loop is detected Algorithm~\ref{algo:dpfromos} terminates with a loop error. Finally, for any set of observations in which all its observations have a connecting neighbor with an appropriate chronological order, and in which no loops occur, Algorithm~\ref{algo:dpfromos} appends a node in a tree data structure, at the corresponding level. Ultimately, the flow tree is transformed to a set of edge sequences. Such edge sequences have the following properties: i) they belong to $E$ as the algorithm queries the corresponding edges in $D$ before adding each node; ii) they are connected, as otherwise the algorithm returns an error; and, finally (iii) they start at host $h$ as the root of flow tree is set to it at Step 2. The previous properties held by the returned elements of Algorithm~\ref{algo:dpfromos} guarantee that the returned sequences of edges are data-paths by definition. Thus, guaranteeing the correctness of the algorithm. Therefore, the following statement holds.

\begin{proposition}\label{stmt:correctness}
Algorithm~\ref{algo:dpfromos} always terminates, and a set of data-paths $\Pi$ is given for a set of observations $\Omega$ that corresponds to a finite and connected network packet traversal at a data-plane $D$, starting at host $h$ (in $D$).
\end{proposition}

Now we turn our attention to constructing the distributed test cases (and eventually the test suites) that guarantee the observation (discovery) of data-paths. Given the Assumption~\ref{ass:traffic}, Proposition~\ref{stm:dpfromobs} implies the following statement:

\begin{corollary}\label{stm:dpsfromtc}
    A set of data-paths $\Pi_D(h,\mathcal{H}(p))$ in a data-plane $D=(H\cup S, E, \mathcal{I})$ is observed iff the distributed test case $tc_D=(h,\mathcal{H}(p))$ is executed.
\end{corollary}

An immediate question follows: is there any distributed test suite that guarantees discovering all data-paths of a given type? Further, is there any distributed test suite that guarantees discovering all data-paths in a data-plane? To reply to these questions, let us first consider how to discover all data-paths for a given traffic type. According to our assumptions (particularly Assumption~\ref{ass:nodestype}), only hosts generate traffic, and therefore data-paths can only start at hosts (see Definition~\ref{def:data-path}). Thus, similarly to Corollary~\ref{stm:dpsfromtc}, provided that all data-paths of type $\tau$ are observed then a test suite $TS_D=\{(h,\mathcal{H}(p))| \mathcal{H}(p)\in\{0,1\}^k \wedge 1_\tau(\mathcal{H}(p))=1 \wedge h\in H\}$, containing a test case for each host in $D$ and for each packet header of type $\tau$ has been executed. Likewise, executing such a test suite guarantees observing all data-paths of type $\tau$. Therefore, the following statement holds.

\begin{proposition}
    All data-paths in the data-plane $D=(H\cup S, E, \mathcal{I})$ (the set $\Pi_D$) are observed iff the distributed test suite $TS_D=\{(h,\mathcal{H}(p))| \mathcal{H}(p)\in\{0,1\}^k \wedge h\in H\}$ is executed.
\end{proposition}

As usual, it is interesting to estimate the length of the distributed test suites which guarantee the discovery of data-paths. We note that the distributed test suite that guarantees discovering all paths of a particular traffic header is a union of all data-paths with that header, and therefore its length is $|H|$. The length of the test suite that guarantees discovering all data-paths of a given type of traffic $\tau$ highly depends on the form of the traffic type indicator functions. If restricting the traffic type indicator functions only to conjunctions in the form $parameter\;i=value\;i$, the length of the test suite of interest is $2^{k-k'}|H|$, where $k'$ is the length of the parameters involved in $1_\tau$. Likewise, the distributed test suite that guarantees discovering all data-plane is  $2^k|H|$. This result sounds unpromising, considering that the length of the parameters involved (the relevant header parameters) are in the order of hundreds. Nonetheless, discovering the data-paths is often focused on \emph{interesting traffic}, i.e., on particular headers or traffic types. In our implementation (see Section~\ref{sec:tool}), we focus also on interesting traffic.

An interesting detail for implementations is the calculation of the maximal number of packets a monitoring system should collect after executing a test case. Perhaps the most interesting is the estimation of an upper bound for the cardinality of a set of data-paths. With this upper bound, the monitoring systems can decide when to stop the monitoring process, after the collection of a number of packets there is a guarantee there exists a loop in the SDN configuration (and therefore no need to continue monitoring).

\begin{proposition}\label{stm:upperbounds}
    The number of finite data-paths $\Pi_D(h,\mathcal{H}(p))$ in a data-plane $D=(V, E, \mathcal{I})$ does not exceed $(|V|-1)^{|V|(|V|-1)}$, and the length of each path does not exceed $|V|(|V|-1)$.
\end{proposition}

\begin{proof}
    A data-path without loops has no repeated edges, otherwise, the path is infinitely repeated. Thus, there exists at most $|V|(|V| - 1)$ distinct edges, and therefore, a longest branch of a path has at most $|V|(|V| - 1) + 1$ nodes. Similarly, a node has at most $|V| - 1$ edges. Therefore, if cloning at each node, there are $(|V|-1)^l$ nodes, where $l$ is the current length of the data-path (starting at $0$). Thus, the number of data-paths is at most $(|V|-1)^{|V|(|V|-1)}$, and their maximal length is $|V|(|V|-1)$. 
\end{proof}

We note that studying the reachability of these upper bounds is an interesting question by itself. However, studying tighter bounds is left for future work.

\section{Experimental results}\label{sec:tool}
In this section, we discuss the developed set of tools implementing the proposed methods, the experimental setup and the obtained experimental results.


\paragraph{Tools description} As previously discussed, our approach is based on distributed test case generation and monitoring. Five different tools (or modules) were developed, those are: (i) the extraction tool - a network packet sniffer which is installed at all network interfaces of the data-plane forwarding devices (POs), that forwards the traffic to the \emph{analyzer tool}; (ii) the packet generation tool - a raw packet generator which is installed at all hosts in the data-plane, that generates packets with specific headers (in the widespread pcap-filter syntax \cite{pcap-filter}) and a Unique Identifier~(UID) as the payload of the packet (to distinguish the test packets); (iii) the analyzer tool - a tool installed at the processing server that receives the packets, filters the packets having the UID and computes the flow trees from the set of observations; (iv) the orchestration tool - receives the traffic of interest to generate, requests the packet generation tool to send a network packet with the requested headers using a specific UID; and finally (v) the User Interface~(UI) tool - a web interface for discovering the data-paths that receives the interesting traffic to discover from the user and reports the resulting data-paths (in the form of a flow tree). In Figure~\ref{fig:toolsinteraction}, we show how the tools interact between each other. The process starts with a user input, asking the UI to generate a packet with a specific header. Then, the UI forwards the request to the orchestrator and based on that header the orchestrator sends the data-plane model to the analyzer, the packets to filter to the extractions agents, and the packets to generate to the generation agents. The generation tool generates the appropriate packet and those packets are sent to the forwarding devices, and eventually the extraction tool captures, and forwards them to the analyzer. The analyzer computes the corresponding paths, and sends the information back directly to the UI, which displays the flow tree.

\begin{figure*}[!htb]
    \centering
\begin{tikzpicture}
    \node[rectangle, draw, rounded corners, fill=white, text=black, very thin, minimum height=0.8cm, minimum width=0.8cm] (gui) {UI};
    \node[] at ([xshift=-2.5cm]gui) (input) {};
    \node[rectangle, draw, rounded corners, fill=white, text=black, very thin, minimum height=0.8cm, minimum width=0.8cm] at ([yshift=-1.5cm]gui) (orch) {Orchestration};
    \node[rectangle, draw, rounded corners, fill=white, text=black, very thin, minimum height=0.8cm, minimum width=0.8cm] at ([yshift=-1cm,xshift=-4cm]orch) (pgen1) {Packet gen. $1$};
    \node[rotate=90] at ([yshift=-1.5cm]pgen1) (pgendots) {\ldots};
    \node[rectangle, draw, rounded corners, fill=white, text=black, very thin, minimum height=0.8cm, minimum width=0.8cm] at ([yshift=-1.5cm]pgendots) (pgenx) {Packet gen. $x$};
    \node[rectangle, draw, rounded corners, fill=white, text=black, very thin, minimum height=0.8cm, minimum width=0.8cm] at ([yshift=-1cm,xshift=4cm]orch) (pextr1) {Extraction $1$};
    \node[rotate=90] at ([yshift=-1.5cm]pextr1) (pextrdots) {\ldots};
    \node[rectangle, draw, rounded corners, fill=white, text=black, very thin, minimum height=0.8cm, minimum width=0.8cm] at ([yshift=-1.5cm]pextrdots) (pextry) {Extraction $y$};
    \node[rectangle, draw, rounded corners, fill=white, text=black, very thin, minimum height=0.8cm, minimum width=0.8cm] at ([xshift=4cm]pextrdots) (analyzer) {Analyzer};

    \path
            (input) edge [thick,->] node[above, align=center] {User Req. \\$\mathcal{H}(p)$}   (gui)
            (gui) edge [thick,->] node[left, align=center] {Discover req. $\mathcal{H}(p)$}   (orch)
            (orch) edge [thick,->] node[left, align=center] {Gen. $\mathcal{H}(p)$,\emph{UID}\\}   (pgen1)
            (orch) edge [thick,->] node[left, align=center] {Gen. $\mathcal{H}(p)$,\emph{UID}}   (pgenx)
            (orch) edge [thick,->] node[right, align=center] {Filter $\mathcal{H}(p)$,\emph{UID}\\}   (pextr1)
            (orch) edge [thick,->] node[right, align=center] {Filter $\mathcal{H}(p)$,\emph{UID}}   (pextry)
            (pextr1) edge [thick,->] node[above, align=center] {Observation\\}   (analyzer)
            (pextry) edge [thick,->] node[below, align=center] {\\Observation}   (analyzer)
            (orch.east) edge [thick,->, bend left=45] node[above, align=center] {Analyze $D$\\}   (analyzer.north)
            (analyzer.east) edge [thick,->, bend right=55] node[right, align=center, above] {Res. $\Pi_D(\mathcal{H}(p))$\\(as a \emph{flow tree})\\}   (gui.east)
            (pgen1) edge [thick,->] node[right, align=center, above] {$p$}   (pextr1)
            (pgen1) edge [thick,->] node[right, align=center, above] {$p$}   (pextry)
            (pgenx) edge [thick,->] node[right, align=center, below] {$p$}   (pextr1)
            (pgenx) edge [thick,->] node[right, align=center, above] {$p$}   (pextry)
    ;
\end{tikzpicture}
    \caption{Tool interaction diagram}\label{fig:toolsinteraction}
\end{figure*}
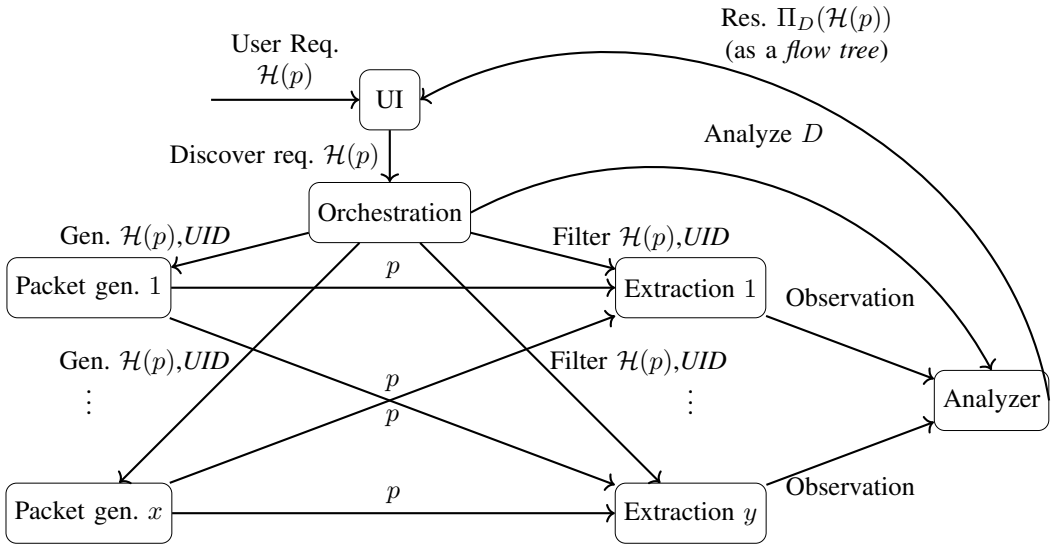

The tools have been developed in different programming languages, including Golang, Python and Javascript. For further details of the implementations, please visit the official repository\footnote{https://github.com/letitbeat/data-path-discovery.}. Likewise, the interested reader can see our short data-path discovery demo\footnote{\label{fn:demo}https://vimeo.com/307046352.}.

\paragraph{Experimental setup} In order to perform our experiments, different data-planes have been simulated using the well-known Mininet simulator, and more precisely, Containernet \cite{containernet} has been used, in order to use docker containers with pre-installed tools. Additionally, the ONOS controller was used for all the experiments. In order to validate our toolset, experiments with installed flow rules were performed. Different data-paths have been configured for different topologies; the flow rules (or intents) were added through the ONOS REST API, the ONOS command line interface or its web (graphical) interface. Our tools are capable of successfully discovering the data-paths of interest (as configured) 100\% of the time. Further, the observed data-paths coincide with the requested ones  (in exception of data-paths containing known bugs for the architecture \cite{ictss18}). An example can be seen in our data-path discovery video demonstration.

\paragraph{Experiments with the ONOS Reactive Forwarding Application} Perhaps the most valuable contribution of our approach is to discover what real applications implement as SDN data-paths. To that end, a well-known SDN application has been installed, namely the ONOS Reactive Forwarding~(ORF) application; Reactive forwarding  ``refers to the mechanism used to install forwarding entries into the network switches - those entries are installed on-demand after a sender starts transmitting packets.'' \cite{reactiveforward}. The application has been installed in our experimental setup, the hosts in the data-plane successfully communicate between them. However, after the discovery process has been executed an unexpected result has been encountered. Indeed, the application does not compute a data-path to forward the network traffic; the application floods the network, i.e., forwarding devices send a received packet to all its neighbors in the data-plane, in exception of the sending device. As an example, consider the data-plane shown in Figure~\ref{fig:data-plane}, a message (with header matching destination TCP port = 22) from $h_1$ to $h_2$ follows the set of data-paths shown in (as a flow tree) Figure~\ref{fig:reactive_observed} while the expected is to have a single data-path, a possibility is shown in Figure~\ref{fig:reactive_expected}.  The previous behavior may be a functional error in the application or a miss-configuration, and in this paper, we do not focus on the cause, however, we highlight the value of the proposed approach in discovering the operational data-paths. The operational data-paths as configured by the ORF application may violate data secrecy, overload the network with unnecessary network packets, and may not conform to the reactive forward specification; both functional and non-functional errors may be detected using our proposed tools and methods.

\begin{figure}[!htb]
    \centering
    \subfloat[t][Operational Data-paths\label{fig:reactive_observed}]
    {
    \begin{tikzpicture}[node distance=1cm,>=stealth',bend angle=45,auto]
    \tikzstyle{switch}=[circle,thick,draw=black!75, inner sep=0.075cm]
    \tikzstyle{host}=[circle,thick,draw=black!75,fill=black!80, text=white, inner sep=0.075cm]
    \tikzstyle{undirected}=[thick]
    \tikzstyle{directed}=[thick,->]
    \node[host] (h1) {$h_1$};
    \node[switch, below of=h1] (s1) {$s_1$};
    \node[switch, below of=h1] (s1) {$s_1$};
    \node[switch] at([yshift=-0.5cm,xshift=-1.1cm]s1) (s2) {$s_2$};
    \node[switch]  at([yshift=-0.5cm,xshift=1.1cm]s1) (s3) {$s_3$};
    \node[host, below left of=s2] (h2) {$h_2$};
    \node[switch, below of=s2] (s4) {$s_4$};
    \node[switch, below right of=s2] (s3_2) {$s_3$};
    \node[host, below right of=s3] (h3) {$h_3$};
    \node[switch, below of=s3] (s2_2) {$s_2$};
    \node[switch, below left of=s3] (s4_2) {$s_4$};
    \node[host, below left of=s4] (h4) {$h_4$};
    \node[switch, below right of=s4] (s3_3) {$s_3$};
    
    \path   (h1)    edge[directed]    node[above] {}   (s1)
            (s1)    edge[directed]    node[above] {}   (s2)
            (s1)    edge[directed]    node[above] {}   (s3)
            (s2)    edge[directed]    node[above] {}   (h2)
            (s2)    edge[directed]    node[above] {}   (s4)
            (s2)    edge[directed]    node[above] {}   (s3_2)
            (s3)    edge[directed]    node[above] {}   (h3)
            (s3)    edge[directed]    node[above] {}   (s4_2)
            (s3)    edge[directed]    node[above] {}   (s2_2)
            (s4)    edge[directed]    node[above] {}   (s3_3)
            (s4)    edge[directed]    node[above] {}   (h4)
    ;
    \end{tikzpicture}
    }
    \vrule\quad
    \subfloat[t][``Expected'' Data-path\label{fig:reactive_expected}]
    {
    \begin{tikzpicture}[node distance=1cm,>=stealth',bend angle=45,auto]
    \tikzstyle{switch}=[circle,thick,draw=black!75, inner sep=0.075cm]
    \tikzstyle{host}=[circle,thick,draw=black!75,fill=black!80, text=white, inner sep=0.075cm]
    \tikzstyle{undirected}=[thick]
    \tikzstyle{directed}=[thick,->]
    
    \node[host] (h1) {$h_1$};
    \node[switch, below of=h1] (s1) {$s_1$};
    \node[switch, below of=s1] (s2) {$s_2$};
    \node[host, below of=s2] (h2) {$h_2$};
    
    \node[] at([xshift=-1.2cm]h1) (spacer1) {};
    \node[] at([xshift=1.2cm]h1) (spacer2) {};
    
    \path   (h1)    edge[directed]    node[above] {}   (s1)
            (s1)    edge[directed]    node[above] {}   (s2)
            (s2)    edge[directed]    node[above] {}   (h2)
    ;
    
    \end{tikzpicture}
    }
    \caption{Observed (operational) vs. expected data-paths of the ONOS reactive forwarding application}
    \label{fig:datapathcomparisson}
\end{figure}
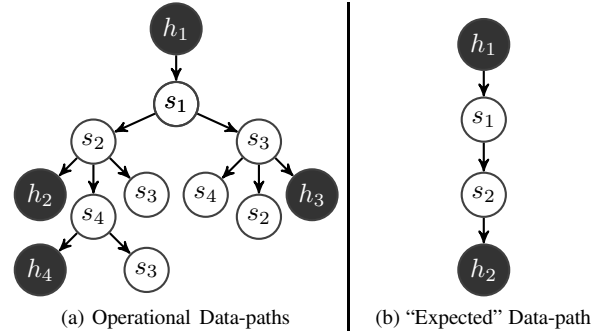

\paragraph{Performance Evaluation}
Although the performance of proposed approach could be evaluated theoretically, performing experiments on our experimental setup can reveal real processing times. In general, after executing a distributed test case, the time required to discover the data-paths is greatly influenced by: (i) the collection of observations and, (ii) the algorithm's (Algorithm~\ref{algo:dpfromos}) processing time. In theory, the required time to collect the observations should be a linear function that depends of the length of a largest data-path. This under the assumption that an observation arrives to the centralized monitor as the network packet traverses the data-plane. Further, the delay between the forwarding devices has the largest influence on processing time. Similarly, the time the algorithm takes to process the set of observations should be a polynomial function that depends on the number of edges in a data-path and the number of data-paths itself (see Proposition~\ref{stm:upperbounds}). 

In order to estimate the execution time of both the data collection and the algorithm processing different topologies have been configured and likewise, different data-paths with different lengths have been configured (up to big data-paths of length $75$). As can be seen in Figure~\ref{fig:perfeval}, the time required to discover data-paths is mostly influenced by the data-collection (and the packet traversal itself in the data-plane). The previous results show that our approach may be applicable, even for identifying a somewhat large number of operational data-paths. Note that the the processing time of the algorithm is negligible in comparison with the data collection time penalty, and therefore, its increase is barely noticeable in the reported plot. 
\begin{figure}[!htb]
    \centering
\begin{tikzpicture}
    \begin{axis}[
        xlabel=Data-path length,
        ylabel=Time (ms),
        legend style={at={(0.02,0.98)}, anchor=north west},
        ytick = {0, 1000, 2000, 3000, 4000, 5000, 6000, 7000, 8000}
        ]

    \addplot[smooth,mark=*, mark size=0.8pt, black] plot coordinates {
        (2,100.976)
        (3,202.306)
        (4,306.47)
        (5,402.6)
        (6,504.384)
        (7,608.43)
        (8,707.935)
        (9,581.710032)
        (10,616.75)
        (11,1031)
        (12,957.084568)
        (13,1212)
        (14,1020)
        (15,1415)
        (16,1519)
        (17,1644)
        (18,1738)
        (19,1819)
        (20,1951)
        (21,2072)
        (22,2132)
        (23,2223)
        (24,2367)
        (25,2439)
        (26,2546)
        (27,2635)
        (28,2774)
        (29,2816)
        (30,2890)
        (31,3016)
        (32,3124)
        (33,3271)
        (34,3352)
        (35,3470)
        (36,3587)
        (37,3658)
        (38,3758)
        (39,3885)
        (40,4009)
        (41,4035)
        (42,4157)
        (43,4290)
        (44,4360)
        (45,4513)
        (46,4539)
        (47,4699)
        (48,4815)
        (49,4894)
        (50,4989)
        (51,5140)
        (52,5169)
        (53,5299)
        (54,5442)
        (55,5540)
        (56,5648)
        (57,5772)
        (58,5729)
        (59,5892)
        (60,6137)
        (61,6090)
        (62,6233)
        (63,6351)
        (64,6456)
        (65,6555)
        (66,6612)
        (67,6829)
        (68,6874)
        (69,6954)
        (70,7062)
        (71,7188)
        (72,7359)
        (73,7681)
        (74,7507)
        (75,7562)
    };
    \addlegendentry{Data Collection Time}
    
    \addplot[smooth,black, mark=*, mark size=0.3pt] plot coordinates {
        (2,32.56)
(3,30.06)
(4,133.15)
(5,84.40)
(6,36.43)
(7,40.44)
(8,55.44)
(9,62.10)
(10,76.50)
(11,92.80)
(12,65.14)
(13,85.72)
(14,89.36)
(15,60.13)
(16,101.45)
(17,94.90)
(18,89.73)
(19,111.49)
(20,75.62)
(21,114.14)
(22,103.71)
(23,105.58)
(24,201.88)
(25,127.46)
(26,126.09)
(27,130.88)
(28,135.82)
(29,121.12)
(30,173.68)
(31,164.53)
(32,149.88)
(33,170.36)
(34,118.08)
(35,95.09)
(36,126.73)
(37,114.11)
(38,140.7239)
(39,95.9316)
(40,117.076)
(41,122.498)
(42,132.3868)
(43,183.88)
(44,152.2038)
(45,132.7318)
(46,120.9921)
(47,133.9612)
(48,193.1049)
(49,166.3023)
(50,147.0791)
(51,144.9401)
(52,224.4102)
(53,197.4109)
(54,187.6447)
(55,195.2918)
(56,207.8479)
(57,156.4889)
(58,176.1167)
(59,163.7641)
(60,143.4527)
(61,181.6429)
(62,173.74)
(63,201.8255)
(64,223.9384)
(65,268.2102)
(66,159.1442)
(67,171.1998)
(68,218.0616)
(69,143.839)
(70,211.5478)
(71,318.1042)
(72,271.3153)
(73,195.5235)
(74,217.5944)
(75,228.5515)
    };
    \addlegendentry{Algorithm's Execution Time}
    \end{axis}

\end{tikzpicture}
\caption{Performance evaluation}
\label{fig:perfeval}
\end{figure}
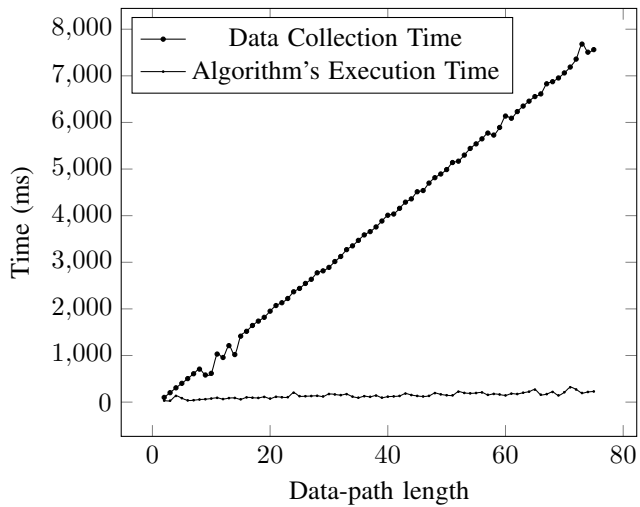

\section{Related work}\label{sec:rel}


The approaches presented in \cite{enase18,ewdts18} aim at performing model-based testing of the whole SDN framework using active testing techniques. After executing the test cases, the authors propose an approach to conclude about the implemented data-paths that inspired this work; their approach is based on specific traffic generation. However, the authors consider only pairwise ICMP echo request/reply traffic generation. Thus, it only discovers the pairwise reachability between hosts, and no other type of data-paths.

In \cite{handigol2012debugger}, Handigol et al. introduce a network debugger tool to help identifying the root cause of network errors and particularly SDN related by showing the sequence of events which lead to aforementioned errors. We believe such tools are extremely important, nonetheless, their tool cannot be directly applied to our problem statement. Likewise, other methods focused on detecting inconsistencies are the ones presented in \cite{veriflow,veridp,zhang2018foces,testpacketgen}.


SDN Traceroute \cite{sdntraceroute} is a tool to trace the path a given packet follows; their approach relies on the prior identification of forwarding devices, then to install a set of rules at each of these devices, in order to forward to the controller the probes sent by the tool. The authors do not focus on providing formal guarantees for discovering all operational data-paths. Additionally, our approach does not install additional forwarding rules (or modifies the packet headers) in order to achieve the data-path discovery. Other approaches which have similar methods (to verify the forwarding rules) are presented in \cite{cherrypick,zhang2017rev}. 

NetSight \cite{netsight} collects and stores information from network packets to build so-called ``packet histories''; later, such information can be used to compute the packet's ``trace''. However, we note that if the test cases are not actively generated / executed, no guarantees can be provided, the conclusion about the data-paths depend on the observed traffic.

An interesting work used to verify certain properties over the data-plane configuration is presented in \cite{mai2011debugging}. The presented tool collects data-plane information and network invariants, and converts them into Boolean expressions. Later the SAT (Boolean satisfiability) problem is used to check that the invariants hold over the data-plane configuration. Again, this approach is not focused on data-path discovery. Another similar work used to verify the data-plane against certain properties is presented in \cite{dobrescu2014software}.


To the best of our knowledge, there are no works that focus on the discovery of data-paths in SDN data-planes which provide guarantees. Furthermore, most of the existing approaches focus on static analysis, which can be ineffective for certain configurations, e.g., forwarding rules to the controller.

\section{Conclusion and future work}\label{sec:conc}
In this paper, an approach for identifying the operational data-paths in SDN data-planes has been presented; the approach is based on distributed test case generation and monitoring. We have shown the conditions when the execution of a test suite (a set of test cases, i.e., network packets to be generated at network hosts) are necessary and sufficient to observe all implemented data-paths. Furthermore, we have developed a set of tools that implement the proposed approach, and experimentally shown that it can be useful for discovering real data-paths installed by SDN applications, and how sometimes, undesired data-paths appear in the data-plane.

This work opens a number of directions that we plan to address in the future. For instance, studying the trade-off between the level of access and the guarantees given. Likewise, it is interesting to study the problem with different assumptions, e.g., regarding the reliability of packet transmission or the synchronized clocks within the forwarding devices. Additionally, it is important to further study the verification that can be performed after the data-paths are discovered; for example, reporting functional errors, repairing the forwarding rule configuration, etc.

\section*{Acknowledgment}

The results obtained in this work were partially funded by the Celtic-Plus European project SENDATE, ID C2015/3-1. The authors would like to thank (in no particular order) Dr. Natalia Kushik from T\'el\'ecom SudParis, as well as the ex-students Takeyuki Koyama and Yohann Lallier from the same institution, and to Prof. Nina Yevtushekno, Dr. Igor Burdonov, Dr. Alexandre Kossatchev from the Russian Academy of Sciences for fruitful discussions that led to the results obtained in this paper.



%

\bibliographystyle{IEEEtran}
\bibliography{references.bib}

\end{document}